\documentclass[twocolumn,pra,showpacs,amsmath,amssymb]{revtex4}

\usepackage{amssymb,amsmath,amsfonts,amsthm,graphicx,color}

\begin{document}

\title{Efficiency of higher dimensional Hilbert spaces for the violation of Bell inequalities}
\author{K\'aroly F.~P\'al}
\email{kfpal@atomki.hu}
\author{Tam\'as V\'ertesi}
\affiliation{Institute of Nuclear Research of the Hungarian Academy of Sciences\\
H-4001 Debrecen, P.O.~Box 51, Hungary}

\def\CC{\mathbb{C}}
\def\RR{\mathbb{R}}
\def\one{\leavevmode\hbox{\small1\normalsize\kern-.33em1}}
\newcommand*{\tr}{\mathsf{Tr}}
\newcommand{\diag}{\mathop{\mathrm{diag}}}
\newtheorem{theorem}{Theorem}[section]
\newtheorem{lemma}[theorem]{Lemma}

\date{\today}

\begin{abstract}
We have determined numerically the maximum quantum violation of over 100
tight bipartite Bell inequalities with two-outcome measurements by each party
on systems of up to four dimensional Hilbert spaces. We have found several cases,
including the ones when each party has only four measurement choices,
when two dimensional systems, i.e., qubits are not sufficient to achieve
maximum violation. In a significant proportion of those cases when qubits are
sufficient, one or both parties have to make trivial, degenerate 'measurements'
in order to achieve maximum violation. The quantum state corresponding to the maximum violation
in most cases is not the maximally entangled one. We also obtain the result,
that bipartite quantum correlations can always be reproduced by measurements
and states which require only real numbers if there is no restriction on the size of
the local Hilbert spaces. Therefore, in order to achieve maximum quantum violation
on any bipartite Bell inequality (with any number of settings and outcomes), there is
no need to consider complex Hilbert spaces.

\end{abstract}

\maketitle

\section{Introduction}\label{intro}

One of the most astonishing features of quantum-mechanics is its nonlocal nature.
Separated observers sharing an entangled state and performing measurements on them may
induce nonlocal correlations which violate Bell inequalities \cite{Bell64}, \cite{CHSH}.
In contrast, separable states satisfy all the possible Bell inequalities with any
measurement settings.

A general setting concerning Bell inequalities is that measurements are made on a
system, which is decomposed into $N$ subsystems. On each of these subsystems one out of
$m_i,\; i=1,\ldots,N$ observables is measured, producing $k_i,\; i=1,\ldots,N$ outcomes
each. In almost all the cases investigated up to now in order to maximally violate them
the dimension of the local state spaces of the shared entangled state did not have to be
larger than the number of outcomes of the respective parties. Some notable exceptions to
it are the bipartite $k_A=3$ and $k_B=2$ Bell inequalities in Ref.~\cite{BG03}, and
families of correlation Bell inequalities with binary outcomes \cite{VP07}, where the
smallest number of measurement settings was found to be $m_A=8$ and $m_B=4$. This latter
case requires states of dimension larger than the number of outcomes to obtain maximal violation.

In the present numerical investigation our aim is two-fold. Firstly, we wish to
demonstrate that by including marginal probabilities in the Bell inequalities it is
further possible to reduce the number of measurement settings. Then we also show that
any bipartite Bell inequality can be violated with settings and states in the real
Hilbert space in the same extent as with settings and states in the complex
Hilbert space.

Actually, we believe that these results are not only of academic interest:
On one hand, higher dimensional systems have been produced in the laboratory in a number of schemes,
subjected to Bell-type tests as well. In particular in Ref.~\cite{Howell} the experimental violation
of a spin-1 Bell inequality has been presented using four-photon states, while in Refs.~\cite{Vaziri},
\cite{Thew} Bell-type tests based on the inequality of Collins et al. \cite{CGLMP} have been performed
for orbital angular momentum and energy-time entangled photons producing qutrits, respectively.
Also, two-photon interference experiments have demonstrated time-bin entanglement up to $d=20$ dimensionality \cite{Riedmatten}.
On the other hand, this investigation can be especially relevant in practical
applications of quantum information protocols. For instance, in quantum cryptography
\cite{Ekert} the key idea is that only local correlations can be created by an
eavesdropper, thus the only useful correlations must have quantum origin. In order
to characterize the set of possible quantum correlations useful for quantum cryptography
applications, it is important to know how effective higher dimensional systems are with
respect to qubits.

In particular, in this paper we considered tight bipartite two-outcome Bell inequalities
corresponding to the facets of the convex polytope \cite{Pitowsky} with up to five settings
2-89 of Ref.~\cite{89list}, and the 31 cases with up to four settings considered by Brunner
and Gisin \cite{BrunGis}. We note that there is some overlap between the two lists. We
used projective measurements in all cases, since for binary outcomes it has been
shown \cite{CHTW04} that general POVM measurements are never relevant. The tools used
in the numerical exploration are gathered in Sec.~\ref{method}, then in Sec.~\ref{disc} we give a list of
tables presenting the numbers corresponding to the maximum quantum violations in
cases of real and complex qubits (3-dimensional spaces), and real qutrits, taking into account degenerate measurements
as well. For all but two inequalities we considered such component spaces were sufficient
for maximum violation. In one case complex qutrits, and in one case real ququarts
(4-dimensional spaces) were necessary to achieve the maximum violation.
For both cases the gain was marginal, not much larger than numerical uncertainty.
The numbers obtained are discussed in Sec.~\ref{disc}, and some conclusions are commented in Sec.~\ref{sum}.
Finally, in Appendix~\ref{app} we provide a proof on the equivalence of real and complex
Hilbert spaces in reproducing bipartite quantum correlations if there is no constraint
on the size of the component Hilbert spaces.


\section{The method}\label{method}
The quantum value of the expression in the Bell inequality is an
expectation value of a Hermitian operator. The maximum expectation value
of such an operator is its largest eigenvalue. Therefore, to find the
maximum quantum violation we have to find those measurement operators
for both Alice and Bob whose combination as it appears in the inequality
gives the largest
possible eigenvalue \cite{FS04}. This way the parameters to be optimized are
those of the measurement operators, no parameter of the vector enters
the problem. The vector can be determined as the eigenvector belonging
to the maximum eigenvalue.

As the outcome of each measurement has to be either 0
or 1, the measurement operators to be considered are projectors in the
component Hilbert spaces of Alice and Bob. In case of 2-dimensional
Hilbert spaces each nondegenerate measurement operator projects to a 1-dimensional
subspace, which may be defined by a unit vector $|m\rangle$ of irrelevant
phase as $|m\rangle\langle m|$. Such a vector can be characterized by
2 parameters, it is convenient to use the two angles on the Bloch-sphere.
As it turned out to be essential, we also considered trivial, degenerate measurement
operators as well. Such a measurement, represented by the zero and the unit operator
brings always the result 0 and 1, respectively.
Obviously, these measurements need not be performed at all, and
the problem becomes equivalent with a smaller one with less measurements.
We performed the optimization with all combinations of nondegenerate, zero and
unit operators.
For 3-dimensional spaces a nondegenerate measurement operator is either
a one or a two-dimensional projector. A unit vector of irrelevant phase
is again sufficient to define either a one and a two-dimensional
projector as $|m\rangle\langle m|$ and $I-|m\rangle\langle m|$,
respectively.  Four real parameters, for example the two polar angles
and the phases of two components (one component may be chosen real)
are needed to characterize such a 3-dimensional complex vector.
Although we have considered only nondegenerate operators, as each of them
may be either a one or a two dimensional projector, many optimization runs
are necessary to cover all combinations.
In the case of 4-dimensional component spaces we confined ourselves
to 2-dimensional projection operators. To make the optimization of the many
parameters involved for all combinations of the dimensions of the
operators would have taken too much computer time. A 2-dimensional
projector in a 4-dimensional complex space requires 8 real parameters to
define.

We may reduce the number of parameters involved by using the
fact that both Alice and Bob may choose their bases freely. With an
appropriate unitary operation we may transform one of the operators, say the
first one, into a diagonal form. This eliminates all parameters of that
operator. Then we may apply another unitary operator that does not
affect the matrix of the first operator to simplify the matrix of the
second operator as much as possible. If there exists further transformation
that leaves the first two matrices unchanged, it may be used to
reduce the number of parameters of the third operator, and so on.
Following this recipe, for qubit spaces the vector
characterizing the first (nondegenerate) operator will be one of the basis vectors
(no parameter), while the one corresponding to the second operator
may be transformed to have both components real (1 parameter).

In a 3-dimensional Hilbert space the components of a unit vector
may be parameterized as $(\cos\varphi\sin\vartheta e^{i\alpha},
\sin\varphi\sin\vartheta e^{i\beta},\cos\vartheta)$, with the 3rd
component is chosen real (4 parameters). The vector corresponding
to the first measurement operator may be transformed to $(0,0,1)$
(no parameter). This form is invariant to a unitary transformation of
the $u_{12}$ type (operation within the subspace spanned by the first
two basis vectors). With such an operation we may eliminate the second
component of the second vector, and we also make its first component
real, leaving the form $(\sin\vartheta_2,0,\cos\vartheta_2)$ (1 parameter).
After this we still have the freedom to eliminate the phase of the
second component of the 3rd vector.

In the case of 4-dimensional
Hilbert spaces, the first measurement operator may be diagonalized
to have the form $diag(1,1,0,0)$.
Then we may apply a further transformation of the form $u_{12}u_{34}$
to simplify the second operator. We can obviously diagonalize
the two $2\times2$ blocks in the upper left and the lower right corners.
Then using the fact that the matrix corresponds to a 2-dimensional
projector, it can be shown that the rest of the transformed matrix
must also have a special form, which with a further allowed operation
may be simplified to the two-parameter form of $(\one+{\cal H})/2$, where
$\one$ is the unit matrix, and
\begin{displaymath}
{\cal H}=
\left(\begin{array}{cccc}
{\hphantom{-}}\cos\phi&{\hphantom{.}}0&{\hphantom{-}}\sin\phi&{\hphantom{.}}0\\
{\hphantom{.}}0&{\hphantom{-}}\cos\psi&{\hphantom{.}}0&{\hphantom{-}}\sin\psi\\
{\hphantom{-}}\sin\phi&{\hphantom{.}}0&-\cos\phi&{\hphantom{.}}0\\
{\hphantom{.}}0&{\hphantom{-}}\sin\psi&{\hphantom{.}}0&-\cos\psi\\
\end{array}\right).
\end{displaymath}
This has been shown in Ref.~\cite{TV06}. The first two matrices leave no
further room to simplify the 3rd and any further operators, it will take
8 parameters to characterize each of them. We have chosen those parameters by
using the fact that the matrix of the most general two-dimensional
projector in the 4-dimensional space may be produced by applying
the most general transformation of the form $u_{12}u_{34}$ to the
two-parameter matrix above. Each of the 2-dimensional unitary operators
$u_{12}$ and $u_{34}$ have 4 parameters. However, an overall phase
is irrelevant, and it also turns out that the effect of the transformation
to the special form will only depend on the difference of two phase
angles in the operators, which makes it possible to eliminate one more
parameter, leaving altogether just the necessary number of $2+(2\cdot 4-2)=8$
parameters.

We determined the maximum violation with both complex and real Hilbert spaces.
A measurement operator in the real space needs just half as many real parameters
to characterize as in a complex space of the same number of dimensions.
The parameters we used were the same as in the complex space with all
phase angles taken to be zero. For optimization we applied an uphill simplex method \cite{NM65}.
As such a method climbs to a local maximum, to find the global one we
restarted the method from random positions many times, at least $10000$ times
for the $4\times 4$ dimensional Hilbert spaces. We still can not
be sure that we have found all global optima, especially for the largest, the
5522 (5 settings of 2-outcome measurements for each of the two parties)
cases. Nevertheless, the results calculated with spaces of different
dimensions are fully consistent with each other.
Either with complex or real spaces, a higher dimensional
calculation has always given at least as large violation as the lower
dimensional ones. When we managed to find a larger value, some optimization
runs still ended up with the lower dimensional result.
From properties of the optimum in the higher
dimensional case, namely the number of terms in the
Schmidt decomposition of the eigenvector and the relation of the subspace
defined by the Schmidt decomposition to the measurement operators
may reveal if it actually corresponds to a lower dimensional case. The
4-dimensional calculations can and do reproduce all lower dimensional
results we considered, including the 2-dimensional cases with degenerate
operators. When the Schmidt decomposition shows that the eigenvector
occupies only 2-dimensional subspaces of Alice and Bob's component spaces,
and there are measurement operators that project to exactly those subspaces,
or to their complementer space, then those measurements for the eigenstate
do behave like degenerate ones. Actually, we realized from such analysis that
in most cases when we found a larger violation with ququarts than with
qubits, the higher dimensionality was not essential, just degenerate operators
had to be considered. In their recent paper Brunner and Gisin also
concluded that for one of their cases they needed degenerate
\cite{BrunGis} measurements.
The 4-dimensional calculation reproduces the 3-dimensional results too,
and may even reveal, which measurement operators should be one, and which ones
should be two-dimensional projectors for maximum violation.

 \begin{table}[tbm]
 \caption{Maximum quantum violation of Bell inequalities calculated with real qubit component spaces, with nondegenerate measurements. Higher dimensional spaces have given no larger violation for these cases. Entries when maximum violation is achieved by the maximally entangled state are marked by stars.}
 \vskip 0.2truecm
 \centering
 \begin{tabular}{l c c c l c c}
 \hline\hline
 Case&Type&Qubit~(R)&\quad\quad&Case&Type&
 Qubit~(R)\\ [0.5ex]
 \hline
 ${\rm CHSH} (A_2)$&2222&0.207107 *&\quad\quad&$A_{27}$&5522&0.648307\hphantom{ *}\\
 $I_{3322} (A_3)$&3322&0.250000 *&\quad\quad&$A_{28}$&5522&0.640314 *\\
 $I^3_{4322}$&4322&0.436492 *&\quad\quad&$A_{30}$&5522&0.569821\hphantom{ *}\\
 $I^2_{4422}$&4422&0.621371\hphantom{ *}&\quad\quad&$A_{31}$&5522&0.573817\hphantom{ *}\\
 $A_5$&4422&0.435334\hphantom{ *}&\quad\quad&$A_{35}$&5522&0.624908\hphantom{ *}\\
 $AS_1$&4422&0.541241 *&\quad\quad&$A_{40}$&5522&0.607864\hphantom{ *}\\
 $AS_2$&4422&0.878493 *&\quad\quad&$A_{42}$&5522&0.619865\hphantom{ *}\\
 $AII_1$&4422&0.605554\hphantom{ *}&\quad\quad&$A_{43}$&5522&0.610765\hphantom{ *}\\
 $AII_2$&4422&0.500000 *&\quad\quad&$A_{51}$&5522&0.660781\hphantom{ *}\\
 $I^5_{4422}$&4422&0.436492 *&\quad\quad&$A_{52}$&5522&0.621861\hphantom{ *}\\
 $I^9_{4422}$&4422&0.461684\hphantom{ *}&\quad\quad&$A_{53}$&5522&0.638610\hphantom{ *}\\
 $I^{10}_{4422}$&4422&0.613946\hphantom{ *}&\quad\quad&$A_{54}$&5522&0.593681\hphantom{ *}\\
 $I^{11}_{4422}$&4422&0.638354\hphantom{ *}&\quad\quad&$A_{57}$&5522&0.660344\hphantom{ *}\\
 $I^{12}_{4422}$&4422&0.618814\hphantom{ *}&\quad\quad&$A_{58}$&5522&0.648890\hphantom{ *}\\
 $I^{17}_{4422}$&4422&0.671409\hphantom{ *}&\quad\quad&$A_{72}$&5522&0.696282\hphantom{ *}\\
 $A_{10}$&5422&0.415390\hphantom{ *}&\quad\quad&$A_{74}$&5522&0.689069\hphantom{ *}\\
 $A_{22}$&5422&0.623457\hphantom{ *}&\quad\quad&$A_{77}$&5522&0.665558\hphantom{ *}\\
 $A_{24}$&5522&0.604799\hphantom{ *}&\quad\quad&$A_{78}$&5522&0.892702\hphantom{ *}\\
 $A_{25}$&5522&0.603379\hphantom{ *}\\
 \hline
 \end{tabular}
 \label{table:realqubit}
 \end{table}

\begin{table}[tbm]
 \caption{Maximum quantum violation is reached with complex
 qubits, no degenerate measurements.}
 \vskip 0.2truecm
 \centering
 \begin{tabular}{l c c c}
 \hline\hline
 Case&Type&Qubit~(R)&Qubit~(C)\\ [0.5ex]
 \hline
 $I^6_{4422}$&4422&0.414214 *&0.449490 *\\
 $I^7_{4422}$&4422&0.441730\hphantom{ *}&0.454837\hphantom{ *}\\
 $A_8$&5422&0.555704 *&0.591650 *\\
 $A_9$&5422&0.451695\hphantom{ *}&0.465243\hphantom{ *}\\
 $A_{11}$&5422&0.445211\hphantom{ *}&0.456108\hphantom{ *}\\
 $A_{12}$&5422&0.452098\hphantom{ *}&0.487709\hphantom{ *}\\
 $A_{15}$&5422&0.447760\hphantom{ *}&0.449628\hphantom{ *}\\
 $A_{19}$&5422&0.588932\hphantom{ *}&0.622630\hphantom{ *}\\
 $A_{20}$&5422&0.564956\hphantom{ *}&0.602240\hphantom{ *}\\
 $A_{23}$&5522&0.528521\hphantom{ *}&0.546073\hphantom{ *}\\
 $A_{26}$&5522&0.486495\hphantom{ *}&0.527555\hphantom{ *}\\
 $A_{29}$&5522&0.456259\hphantom{ *}&0.492064\hphantom{ *}\\
 $A_{32}$&5522&0.396861\hphantom{ *}&0.413553\hphantom{ *}\\
 $A_{33}$&5522&0.561909\hphantom{ *}&0.622631\hphantom{ *}\\
 $A_{36}$&5522&0.419088\hphantom{ *}&0.438868\hphantom{ *}\\
 $A_{37}$&5522&0.456106\hphantom{ *}&0.486887\hphantom{ *}\\
 $A_{38}$&5522&0.428958\hphantom{ *}&0.469913\hphantom{ *}\\
 $A_{39}$&5522&0.612269\hphantom{ *}&0.617203\hphantom{ *}\\
 $A_{41}$&5522&0.419234\hphantom{ *}&0.478563\hphantom{ *}\\
 $A_{47}$&5522&0.402679\hphantom{ *}&0.460854\hphantom{ *}\\
 $A_{48}$&5522&0.431439\hphantom{ *}&0.454841\hphantom{ *}\\
 $A_{49}$&5522&0.454198\hphantom{ *}&0.466694\hphantom{ *}\\
 $A_{50}$&5522&0.500000 *&0.518290\hphantom{ *}\\
 $A_{79}$&5522&0.606128\hphantom{ *}&0.624315\hphantom{ *}\\
 $A_{81}$&5522&0.662368\hphantom{ *}&0.669010\hphantom{ *}\\
 $A_{83}$&5522&0.696038\hphantom{ *}&0.696166\hphantom{ *}\\
 $A_{85}$&5522&0.610060\hphantom{ *}&0.641141\hphantom{ *}\\
 $A_{86}$&5522&0.780438\hphantom{ *}&0.800443\hphantom{ *}\\
 \hline
 \end{tabular}
 \label{table:complexqubit}
 \end{table}

 \begin{table}[tbm]
 \caption{Maximum quantum violation is reached with real
 qubits, with some measurement operators degenerate.}
 \vskip 0.2truecm
 \centering
 \begin{tabular}{l c c c c}
 \hline\hline
 Case&Type&Qubit~(R)&Qubit~(C)&Qubit~(R)\\
 &&nondeg.&nondeg.&deg.~op.\\ [0.5ex]
 \hline
 $I^1_{4322}$&4322&0.154701\hphantom{ *}&0.236068\hphantom{ *}&0.414214 *\\
 $I^2_{4322} (A_4)$&4322&0.231613\hphantom{ *}&0.259587\hphantom{ *}&0.299038 *\\
 $A_6$&4422&0.222941\hphantom{ *}&0.232051 *&0.299038 *\\
 $I^3_{4422}$&4422&0.238042\hphantom{ *}&0.238042\hphantom{ *}&0.414214 *\\
 $I^4_{4422}$&4422&0.055979\hphantom{ *}&0.055979\hphantom{ *}&0.414214 *\\
 $I^{13}_{4422}$&4422&0.249466\hphantom{ *}&0.250000 *&0.434855\hphantom{ *}\\
 $I^{14}_{4422}$&4422&0.407621\hphantom{ *}&0.410296\hphantom{ *}&0.479410 *\\
 $I^{15}_{4422}$&4422&0.238273\hphantom{ *}&0.250000 *&0.434855\hphantom{ *}\\
 $I^{16}_{4422}$&4422&0.240659\hphantom{ *}&0.240659\hphantom{ *}&0.414214 *\\
 $A_{17}$&5422&0.221946\hphantom{ *}&0.221946\hphantom{ *}&0.375447\hphantom{ *}\\
 $A_{18}$&5422&0.210377\hphantom{ *}&0.212229\hphantom{ *}&0.384355\hphantom{ *}\\
 $A_{34}$&5522&0.461083\hphantom{ *}&0.513972\hphantom{ *}&0.535012\hphantom{ *}\\
 $A_{44}$&5522&0.500000 *&0.533925\hphantom{ *}&0.536494\hphantom{ *}\\
 $A_{55}$&5522&0.451941\hphantom{ *}&0.486823\hphantom{ *}&0.621320 *\\
 $A_{56}$&5522&0.675426 *&0.675426 *&0.689312 *\\
 $A_{59}$&5522&0.430220\hphantom{ *}&0.430220\hphantom{ *}&0.448826\hphantom{ *}\\
 $A_{63}$&5522&0.327627\hphantom{ *}&0.327627\hphantom{ *}&0.479410 *\\
 $A_{69}$&5522&0.330388\hphantom{ *}&0.330388\hphantom{ *}&0.609610\hphantom{ *}\\
 $A_{70}$&5522&0.465198\hphantom{ *}&0.465198\hphantom{ *}&0.605223\hphantom{ *}\\
 $A_{71}$&5522&0.418729\hphantom{ *}&0.418729\hphantom{ *}&0.449016\hphantom{ *}\\
 $A_{73}$&5522&0.800326\hphantom{ *}&0.852797\hphantom{ *}&0.883138\hphantom{ *}\\
 $A_{75}$&5522&0.572736\hphantom{ *}&0.587052\hphantom{ *}&0.605151\hphantom{ *}\\
 $A_{80}$&5522&0.136376\hphantom{ *}&0.174354\hphantom{ *}&0.375447\hphantom{ *}\\
 $A_{82}$&5522&0.314943\hphantom{ *}&0.314943\hphantom{ *}&0.454573\hphantom{ *}\\
 $A_{84}$&5522&0.605340\hphantom{ *}&0.619437\hphantom{ *}&0.623457\hphantom{ *}\\
 $A_{88}$&5522&0.076842\hphantom{ *}&0.076842\hphantom{ *}&0.414214 *\\
 \hline
 \end{tabular}
 \label{table:realqubitd}
 \end{table}

 \begin{table}[tbm]
 \caption{Maximum quantum violation is reached with complex
 qubits, with some measurement operators degenerate.}
 \vskip 0.2truecm
 \centering
 \begin{tabular}{l c c c c c}
 \hline\hline
 Case&Type&Qubit~(R)&Qubit~(C)&Qubit~(R)&Qubit~(C)\\
 &&nondeg.&nondeg.&deg.~op.&deg.~op.\\ [0.5ex]
 \hline
 $A_{16}$&5422&0.416036\hphantom{ *}&0.416036\hphantom{ *}&0.446167\hphantom{ *}&0.457107 *\\
 $A_{45}$&5522&0.482065\hphantom{ *}&0.509936\hphantom{ *}&0.534037\hphantom{ *}&0.537239\hphantom{ *}\\
 $A_{61}$&5522&0.307654\hphantom{ *}&0.307654\hphantom{ *}&0.395168\hphantom{ *}&0.401925\hphantom{ *}\\
 $A_{62}$&5522&0.219048\hphantom{ *}&0.231812\hphantom{ *}&0.395168\hphantom{ *}&0.401925\hphantom{ *}\\
 $A_{66}$&5522&0.345116\hphantom{ *}&0.360817\hphantom{ *}&0.452098\hphantom{ *}&0.487709\hphantom{ *}\\
 \hline
 \end{tabular}
 \label{table:complexqubitd}
 \end{table}

 \begin{table*}[tbm]
 \caption{Maximum quantum violation is reached with real qutrits.}
 \vskip 0.2truecm
 \centering
 \begin{tabular}{l c c c c c c}
 \hline\hline
 Case&Type&Qubit~(R)&Qubit~(C)&Qubit~(R)&Qubit~(C)&Qutrit~(R)\\
 &&nondeg.&nondeg.&deg.~op.&deg.~op.&\\ [0.5ex]
 \hline
 $I^1_{4422} (A_7)$&4422&0.197048\hphantom{ *}&0.197048\hphantom{ *}&0.250000 *&0.250000 *&0.287868\hphantom{ *}\\
 $I^8_{4422}$&4422&0.420651\hphantom{ *}&0.420651\hphantom{ *}&0.484313 *&0.484313 *&0.487768\hphantom{ *}\\
 $I^{18}_{4422}$&4422&0.181236\hphantom{ *}&0.181236\hphantom{ *}&0.543599\hphantom{ *}&0.543599\hphantom{ *}&0.642967\hphantom{ *}\\
 $I^{19}_{4422}$&4422&0.369700\hphantom{ *}&0.430724 *&0.443587\hphantom{ *}&0.443587\hphantom{ *}&0.497171\hphantom{ *}\\
 $I^{20}_{4422}$&4422&0.305645\hphantom{ *}&0.305645\hphantom{ *}&0.434324\hphantom{ *}&0.434324\hphantom{ *}&0.449669\hphantom{ *}\\
 $A_{13}$&5422&0.397412\hphantom{ *}&0.403098\hphantom{ *}&0.414214 *&0.414214 *&0.419982\hphantom{ *}\\
 $A_{14}$&5422&0.449958\hphantom{ *}&0.453901\hphantom{ *}&0.452465\hphantom{ *}&---&0.464584\hphantom{ *}\\
 $A_{46}$&5522&0.446602\hphantom{ *}&0.449849\hphantom{ *}&---&---&0.458105\hphantom{ *}\\
 $A_{60}$&5522&0.252968\hphantom{ *}&0.252968\hphantom{ *}&0.375447\hphantom{ *}&0.375447\hphantom{ *}&0.390611\hphantom{ *}\\
 $A_{64}$&5522&0.375234\hphantom{ *}&0.375234\hphantom{ *}&0.375447\hphantom{ *}&0.375447\hphantom{ *}&0.390089\hphantom{ *}\\
 $A_{65}$&5522&0.208545\hphantom{ *}&0.208545\hphantom{ *}&0.347759 *&0.353146\hphantom{ *}&0.355021\hphantom{ *}\\
 $A_{67}$&5522&0.395696\hphantom{ *}&0.395696\hphantom{ *}&---&---&0.396289\hphantom{ *}\\
 $A_{68}$&5522&0.385731\hphantom{ *}&0.385731\hphantom{ *}&---&---&0.395718\hphantom{ *}\\
 $A_{76}$&5522&0.404741\hphantom{ *}&0.415397\hphantom{ *}&0.447555\hphantom{ *}&0.447555\hphantom{ *}&0.489863\hphantom{ *}\\
 $A_{89}$&5522&0.131420\hphantom{ *}&0.131420\hphantom{ *}&0.250000 *&0.250000 *&0.288932\hphantom{ *}\\
 \hline
 \end{tabular}
 \label{table:realqutrit}
 \end{table*}

\section{Discussion of the results}\label{disc}

We calculated the maximum violation of the tight bipartite Bell inequalities
$A_2-A_{89}$ listed in Ref.~\cite{89list} ($A_1$ is a trivial 1122 type, which
can not be violated).
These inequalities are the part involving at most 5 measurement settings per party
of a huge list of inequalities obtained with the method described in
Ref.~\cite{89listgen}. We also included the 31 known tight inequalities with
up to 4 measurement settings per party considered recently by Brunner and
Gisin \cite{BrunGis}. We adopted the notation used in that paper.
Out of the 26 inequalities of 4422 type, 20 was newly introduced there, while
$I^1_{4422}$ was presented in \cite{14422}, $I^2_{4422}$ in \cite{24422},
$A_5$, $A_6$, $AII_1$ and $AII_2$ in \cite{A5A6AII12}, while $AS_1$ and $AS_2$
in \cite{AS12}. The only 2222 one is the Clauser-Horne-Shimony-Holt (CHSH) inequality
\cite{CHSH}. The Bell inequality found in Ref.~\cite{PS01} is the only 3322 type, and
the three 4322 cases were introduced in Ref.~\cite{14422}.
The two lists we considered have some overlap, we marked those cases in our
tables. For every inequality in the lists the classical value to be violated is
zero, except for $I^7_{4422}$, where it is one. The maximum violations we
show in the tables are just the maximum eigenvalues we found, except for the
case $I^7_{4422}$, where it is one less.

In Table~\ref{table:realqubit} we listed all those cases for which we could not find
a stronger violation in any of our calculations than the maximum violation we
achieved with real qubits, performing only nondegenerate measurements.
In all tables we marked with a star the cases when maximum violation was
achieved with the maximally entangled state. For most instances this is not so,
which has also been noted in Ref.~\cite{BrunGis}.
Table~\ref{table:complexqubit} contains the inequalities when we got the
maximum violation with measurements on complex qubits. For the cases in these
tables we got the same values for maximum violation with complex qutrits and complex
ququarts than with complex qubits, and real qutrits did as well as
real qubits. However, with real ququarts we could always
achieve the same amount of violation as with complex qubits.
It is generally true that if a bipartite Bell inequality with arbitrary outputs per party
can be violated by a certain amount with projective measurements in $n$-dimensional Hilbert spaces, than they can
be violated by at least as much with projective measurements in $2n$-dimensional
real Hilbert spaces. This property is an immediate outcome of an even more general statement, which is
provided in Appendix~\ref{app}. It is an open question, whether Lemma~\ref{lemma:main} could be somehow
generalized so that this statement would be true for any multipartite Bell inequalities as well.
From the construction it follows, and we have demonstrated in Appendix~\ref{app}, that
the Schmidt-decomposition of the state in the 4-dimensional real space has 4 terms,
the Schmidt-coefficients are pairwise equal, and the ratio of the pairs
equals to the ratio of the Schmidt-factors from the qubit case with the
same violation.

There are surprisingly many inequalities that can be violated more, sometimes
very significantly more by allowing measurements to be degenetate, than by
confining ourselves only to nontrivial ones.
Table~\ref{table:realqubitd} and Table~\ref{table:complexqubitd} show the cases
when we got the maximum violation with real and complex qubits, respectively,
taking one or more measurements of Alice, or Bob, or of both of them degenerate,
i.e., either unity or zero. As we have already mentioned, the four-dimensional
calculations can always reproduce these values even by confining ourselves to
rank 2 measurements (2-dimensional projectors) by operators that project
onto the subspace the eigenvector occupies, or onto the orthogonal one.
However, when a complex qubit result is reproduced with real ququarts,
the eigenvector requires the whole component spaces
(4 terms in Schmidt decomposition),
therefore effect of degenerate operators can not be simulated with rank 2
operators this way.

Brunner and Gisin \cite{BrunGis} calculated the maximum quantum violation
by applying degenerate measurements only for their $I^4_{4422}$ inequality.
They did that after realizing that this inequality can not be violated by
the maximally entangled state without such measurements.
They state $(1/\sqrt{2}-1/2)$ as the value of maximum quantum violation, which they
achieved by taking two measurement operators of both parties degenerate.
We found twice as large maximum violation by taking two measurement operators
of only one party degenerate (see Table~\ref{table:realqubitd}). We also found
that a very small violation may be achieved by using only true two-outcome
measurement. The violating state is far from the maximally entangled state,
it has Schmidt coefficients of 0.9158 and 0.4016.

So far we have only shown cases for which maximum violation could be achieved
in qubit spaces. The existence of Bell inequalities for which this is not the case has been proved
in Refs.~\cite{VP07,Brunner,Perez}. Particularly, in Ref.~\cite{VP07} we were able to give concrete
examples of correlation Bell inequalities (i.e., inequalities without local marginals) whose maximal
violation is not achieved by qubits.
In the present list we found numerically quite a few such cases, now for Bell expressions with marginals.
In all such cases except for two, real qutrit spaces were enough for maximum violation, see
Table~\ref{table:realqutrit}. For most of them, in two dimensions larger violation
can be achieved by allowing degenerate operators than by not allowing them
(no entry in the appropriate place, when it is not so). With qutrits we can do
even better. However, for most entries in the list the increase is quite
small, no more than a couple of percents, sometimes even much less, which
means these cases may have no practical and experimental relevance. For a few cases the gain is
more than 10\%. We find the largest increase (about 0.1, or 18\%) for $I^{18}_{4422}$.
It is interesting to note that there exist Bell
inequalities that can be violated more with real qutrits than with complex
qubits, and there are also examples for the opposite (at least without allowing
degenerate measurements for qutrits, which we have not tried). For all cases in
Table~\ref{table:realqutrit}
each party has at least 4 measurements, in the smallest ones each of them has just
4. We will show in a forthcoming publication that for correlation type inequalities
to get larger violation with higher-dimensional spaces than with qubits, one of the
parties must have at least 4 measurements, and then the other one must have at least
7 measurements. All 4422, 5422 and 6422 correlation type Bell inequalities
can maximally be violated by qubits.

We found one single inequality in the list that we could violate more with complex qutrits than
with real ones or with qubits. The maximum violation of $A_{21}$ (5422) with real qubits
(no degenerate measurement) is 0.099090, with complex qubit (no degenerate measurement) is
0.125000, with real and complex qubit (degenerate measurement allowed) 0.299038
(maximally entangled state), with real qutrit 0.316523, and with complex
qutrit 0.317496. The last improvement is absolutely marginal, but it does not seem
to be due to numerical error.

For $A_{87}$ (5522) we found we need ququarts to get maximum violation, but
the improvement was even less convincing. The best qubit value is 0.756199
(both with real and complex qubit), while the maximum we got with
both real and complex ququarts is 0.756247. From a more detailed
analysis of the solution we could not see a way to reduce it to
a lower dimensional space. It turned out that this violation could be achieved by
taking two measurement operators equal. Therefore, we calculated the maximum
violation with qubits of the 5422 inequality we got by uniting these two measurements,
and we found 0.755931, a slightly smaller value than for the original inequality.
The difference from the ququart value is still extremely small, but at least it
seems to be more than numerical error.

In our calculations the maximum number of dimensions for the component spaces
were four. Moreover, we allowed degenerate measurements only for qubit
spaces, and confined ourselves to rank 2 measurements in four dimensions.
For some cases on the list it is possible, that without these restrictions one
could find a larger maximum quantum violation.

\section{Summary}\label{sum}

Let us briefly summarize the main results achieved in this work.

We investigated numerically the maximum values on tight bipartite two-outcome Bell inequalities in cases when
the local Hilbert space was restricted to $d=2,3,4$ dimensions. We found Bell inequalities with four measurement settings for each side where qutrits were needed to achieve maximal violation, and with five measurement settings for each side where ququarts were needed to achieve maximal violation. We may interpret these results via the concept of witnessing the Hilbert space dimension \cite{Brunner,Perez}. The question is that given a joint probability distribution of measurement results performed by separate parties, is it possible to set a bound on the dimension of the multipartite state space? Thus, dimension witnesses are operators \cite{Brunner}
able of bounding the dimension of a quantum system. This allows one to test experimentally the size of the underlying Hilbert space, which otherwise is a rather abstract concept. Therefore, by adapting this language, we can say that we found numerically tight Bell inequalities which act as dimension witnesses for qubits and qutrits.

On the other hand, in analogy to the terminology {\it dimension} witnesses one may inquire whether {\it reality} witnesses could be constructed, which would be able to distinguish complex Hilbert spaces from real Hilbert spaces. Actually, the existence of such kind of a witness has been quested by Gisin in Ref.~\cite{AS12}. However, according to our result presented in Appendix~\ref{app}, we may safely say that reality witness cannot be constructed for the case of two parties since by doubling the size of the local complex Hilbert space of each party one may reconstruct all the joint probabilities with local real Hilbert spaces as well. Although, the question is remained open for multipartite systems, numerical study supports us to believe that our Lemma holds for the most general case as well.

\acknowledgements
The authors thank Antonio Ac\'in, Nicolas Brunner, and Robert Englman for valuable discussions.

\appendix

\section{On the equivalence of real and complex Hilbert spaces in reproducing
bipartite quantum correlations}\label{app}

Here the following main result is shown:
\begin{lemma}
Joint probabilities between two separated observers which has quantum origin can always be
reproduced by measurements and states which require only real numbers.
\label{lemma:main}
\end{lemma}

This fact which is interesting by its own, has some striking consequences, an immediate one
is that the maximum quantum violation of any bipartite Bell inequality (with any number of
settings and outcomes) can be achieved in the real Hilbert space as well.

To set the scene, we assume that two separated observers, Alice and Bob, may perform one of a finite number of
measurements, and that each measurement has a certain number of outcomes. We label outcomes corresponding to different measurements distinctly, so that each outcome $a$ and $b$ is uniquely associated to a single measurement of Alice and Bob, respectively. Let $S_A$ and $S_B$ be $n$-dimensional complex Hilbert spaces of the two parties, respectively, and $|V\rangle$ be any vector in the tensor product space $S_A\otimes S_B$. Let $P_a$ ($P_b$) be projection operator associated with outcome $a$ ($b$) of $S_A$ ($S_B$).

In the light of the above definitions, we say that the joint probabilities $p_{ab}$
admit a quantum representation \cite{NPA07} if there exists a quantum state $\rho$ on the composite Hilbert
space, a set of projectors $P_{a}\otimes \one$ of Alice's and a set of projectors $\one \otimes P_{b}$ of Bob's system, such that
\begin{equation}
p_{ab}=Tr(P_{a}P_{b}\rho).
\label{joint}
\end{equation}
Note, that since we do not impose any limitation on the dimension of the local Hilbert spaces, we may
consider projection operators instead of the more general POVM measurements. The Bell expression consists of
a linear combination of probabilities~(\ref{joint}). The projectors belonging to different outcomes
of a measurement are orthogonal to each other, and they sum up to unity.

First we prove the following correspondence between joint distributions arising from projection
measurements in complex $n$-dimensional local Hilbert spaces and projection measurements
in real $2n$-dimensional local Hilbert spaces:

\begin{lemma}
There exist projection operators $P'_a$ and $P'_b$ of the $2n$-dimensional
real spaces $S'_A$ and $S'_B$, respectively, and
$|V'\rangle\in S'_A\otimes S'_B$ such that the corresponding
expectation values are equal,
i.e.,
\begin{equation}
\langle V|P_a\otimes P_b|V\rangle=\langle V'|P'_a\otimes P'_b|V'\rangle,
\label{expect}
\end{equation}
where the state $|V\rangle$ and operators $P_a$, $P_b$ are defined above,
and $|V'\rangle$, $P'_a$ and $P'_b$ depend only on $|V\rangle$, $P_a$ and $P_b$,
respectively.

\end{lemma}

\begin{proof}
Let us use a matrix representation. Let us choose orthonormal bases in each
component space, and let the basis in the product space be the basis
consisting of the products of the basis vectors of the component spaces.
Hence, we can write,
\begin{equation}
|V\rangle = \sum V_{ij}|v^A_i\rangle|v^B_j\rangle,
\label{V}
\end{equation}
and
\begin{align}
A_{ij} &=\langle v^A_i| P_a|v^A_j\rangle, \\
B_{ij} &=\langle v^B_i| P_b|v^B_j\rangle,
\end{align}
where the basis vectors
$\{|v^A_i\rangle\}_{i=1}^n$ and $\{|v^B_j\rangle\}_{j=1}^n$ span
respectively Alice and Bob's local state spaces.
This way the vectors of the product space will be represented by
matrices of two indices. Then the expectation value above can be
expressed as
\begin{displaymath}
\sum_{i,j,k,l}V^*_{ij}A_{ik}B_{jl}V_{kl}=Tr(A V B^T V^{\dagger})
\label{expectmat}
\end{displaymath}
where $A$, $B$ and $V$ are the matrix representations of
$P_a$ and $P_b$ and $|V\rangle$, respectively. The value of the expression
is a real number, as it gives the expectation value of a Hermitian operator
in the product space.

Let us consider the following mapping \cite{Myrheim}. Let us replace each component $v_i=v_i^R+iv_i^I$
of the $n$-dimensional complex vector with the two-element real block of $(v_i^R, v_i^I)$,
and each component $A_{ij}=A^R_{ij}+iA^I_{ij}$ of a two-index matrix with the $2\times2$ block of
\begin{displaymath}
\left(
\begin{array}{cc}
A_{ij}^R&-A^I_{ij}\\
A^I_{ij}&A^R_{ij}\\
\end{array}\right).
\end{displaymath}
One can prove that the image of the product of either a matrix and a vector, or
two matrices will be equal to the corresponding product of the images. For $n=1$
this is easy to show. For $n>1$ the multiplication in the $2n$-dimensional space may
be done block-by-block, yielding the correct result. The mapping also conserves the
linear combinations of both vectors and matrices. When transposing matrices one has to be
careful. The image of the transpose of a matrix will be the transpose of the image of the
complex conjugate of the matrix. The complex conjugation is needed to get the
$2\times2$ blocks right (they are not transposed in the image of the transpose).
Hermitic conjugation is preserved by the mapping.
It is also easy to see that the trace operation on the image will give a real
number, which is twice the real part of the value calculated for the original
complex matrix (in each block the real part of the diagonal matrix element
will occur twice, while the imaginary part will be off-diagonal).
Given these rules in hand it is clear that the image of a projector is
also a projector, the images of orthogonal projectors are orthogonal projectors,
and if matrices sum up to unity, their images will do so, too. Therefore, the images
of a set of measurement operators will satisfy the properties required.

Let $|V'\rangle$ be the vector in $S'_A\otimes S'_B$ whose matrix $V'$ is constructed
with the above rule for 2-index matrices from the matrix $V$ of $|V\rangle$, and then
multiplied by $1/\sqrt{2}$ to get it properly normalized. We note
that the mapping rule to be applied in the product space is not the same as
the one applied in the component spaces. That rule would actually give just $2n^2$
components instead of the $(2n)^2$ ones. Let there be the matrix of $P'_a$ and
$P'_b$, i.~e.\ $A'$ and $B'$ the image of $A$ and $B^*$, respectively.
Then $A' V' B'^T V'^\dagger$ will be the image of $(1/2)A V B^T V^\dagger$, the factor
of $1/2$ is occurring due to
the $1/\sqrt{2}$ normalization factor in the construction of $V'$ from $V$.
As the trace of $A V B^T V^\dagger$, which is the expectation value in the complex space, is real,
its value is one half of the trace of its image, i.e., it is equal to the trace
of $A' V' B'^T V'^\dagger$, which is the expectation value in the real space.

\end{proof}

Note that for an arbitrary mixed state $\rho=\sum\lambda_i |V_i\rangle\langle V_i|$ the
expectation value
$Tr(P_a P_b\rho)$ is the convex sum of the expectations~(\ref{expect}) with coefficients
$\lambda_i$, which entails the main result Lemma~\ref{lemma:main} we wanted to show.

Aside from its conceptual interest, we mention two interesting situations where this fact
may prove to be useful beyond justifying our numerical experience that real ququarts
could yield at least the same amount of violation as complex qubits.

On one hand, in the inequality presented by Bechmann-Pasquinucci and Gisin in
Ref.~\cite{BG03} having three and two
measurement outcomes per Alice and Bob, respectively, the maximum quantum violation
can be achieved with
projective measurements sharing a maximally entangled state of dimension $3$.
However, numerical evidence
suggests that using measurement settings which require real numbers, the optimum
quantum violation could not be reached. It arisen as a natural question \cite{AS12}
whether a higher value could be achieved by using only real numbers but allowing to
occupy larger Hilbert spaces. Our result gives the answer in negative for this
question regarding this particular Bell inequality and also prove conclusively
that all bipartite Bell inequalities can be maximally violated by quantum states
and measurement settings which need in an appropriate basis only real numbers. This
latter problem for the general multipartite case was posed by Gisin (see also problem 32,
fundamental questions number 11 in Ref.~\cite{open}).

On the other hand, in Ref.~\cite{NPA07} a hierarchy of conditions has been formulated
through a semidefinite program \cite{VB96}. This approach can be used for instance to
obtain upper bounds on the quantum violation of arbitrary Bell inequalities. In this
case, however the matrix $\Gamma$ in question, which should satisfy the positive
semidefinite constraint is in general Hermitian. Our results, however, entails that
this matrix needs to be in fact real valued, i.e., must be a symmetric matrix. This
stronger condition thus may provide us with a tighter upper bound on any bipartite
Bell inequality, than the one which originally required the weaker Hermitian condition.

Now let us illustrate with a simple example, consisting of a qubit at each party,
the method how to obtain the projection operators and the respective states from
the original complex valued ones. In this case the state of two qubits can be written
in an appropriate basis as
$|V\rangle =\alpha |v^a_1\rangle|v^b_1\rangle + \beta |v^a_2\rangle|v^b_2\rangle$,
where the $\alpha$ and $\beta$ Schmidt coefficients are non-negative numbers, their
square adding up to $1$. Thus the matrix $V$ in Eq.~(\ref{V}) takes the following
simple form,
$diag(\alpha,\beta)$ whereas a non-degenerate projector on the state space of Alice
and Bob can be written as
$P_{\nu}=(\one \pm \vec \nu \vec\sigma)/2,\;\;\nu\in a,b$. Applying the mapping rule,
discussed above, we obtain the following real valued $4\times4$ matrices,
$V=(1/\sqrt 2)diag(\alpha,\alpha,\beta,\beta)$, implying the entangled state
(with nonzero $\alpha$ and $\beta$) in the 4-dimensional state space,
$|V\rangle=(\alpha|00\rangle + \alpha|11\rangle + \beta|22\rangle + \beta|33\rangle)/\sqrt 2$
and the corresponding projection operators
$P'_{\nu}=(\one \pm \vec \nu \vec\sigma')/2,\;\;\nu\in a,b$, where
$\sigma'_x=\sigma_x\otimes\one$,
$\sigma'_y=-\sigma_y\otimes\sigma_y$ and $\sigma'_z=\sigma_z\otimes \one$.

\end{document}